\documentclass[11pt, reqno]{amsart}
\usepackage[margin=3.6cm]{geometry}
\usepackage{hyperref}
\usepackage{array,booktabs,tabularx}
\usepackage{graphicx}
\usepackage{amssymb}
\usepackage{enumerate}
\usepackage{color}
\usepackage{esint}
\usepackage{mathtools}
\usepackage{dsfont}
\usepackage{ esint }
\usepackage{enumitem}
\usepackage{wrapfig}
\usepackage{caption}
\usepackage{subcaption}

\newtheorem{theorem}{Theorem}

\newtheorem{proposition}[theorem]{Proposition}

\newtheorem{Lem}[theorem]{Lemma}

\newtheoremstyle{named}{}{}{\itshape}{}{\bfseries}{.}{.5em}{\thmnote{#3 }#1} \theoremstyle{named} 

\theoremstyle{definition}
\newtheorem{definition}{Definition}
\newtheorem{remark}{Remark}

\newcommand{\leb}{\lambda}

\newcommand{\R}{{\mathbb R}}

\newcommand{\ep}{\varepsilon}

\newcommand{\gives}{\ensuremath{\rightarrow}}

\newcommand{\setst}[2]{\ensuremath{ \{ #1\,|\,#2 \}}}

\newcommand{\abs}[1]{\ensuremath{\left| #1 \right|}}
\newcommand{\lr}[1]{\ensuremath{\left(#1 \right)}}
\newcommand{\norm}[1]{\left\lVert#1\right\rVert}
\newcommand{\inprod}[2]{\ensuremath{\left\langle#1,#2\right\rangle}}

\newcommand{\w}{\omega}
\newcommand{\dell}{\ensuremath{\partial}}
\newcommand{\set}[1]{\ensuremath{\{#1\}}}

\def\XXint#1#2#3{{\setbox0=\hbox{$#1{#2#3}{\int}$} \vcenter{\hbox{$#2#3$}}\kern-.5\wd0}}

\DeclareMathOperator{\vol}{vol}

\newcommand{\N}{\mathbb N}

\newcommand*\pFqskip{8mu}
\catcode`,\active
\newcommand*\pFq{\begingroup
        \catcode`\,\active
        \def ,{\mskip\pFqskip\relax}%
        \dopFq
}
\catcode`\,12
\def\dopFq#1#2#3#4#5{%
        {}_{#1}F_{#2}\biggl[\genfrac..{0pt}{}{#3}{#4};#5\biggr]%
        \endgroup
}

\DeclareMathOperator{\HO}{HO}
\DeclareMathOperator{\Op}{P}
\DeclareMathOperator{\Spec}{Spec}

\title[Nodal Sets and Level Spacings]{Level Spacings and Nodal Sets at Infinity for Radial
  Perturbations of the Harmonic Oscillator}

\author[T. Beck]{Thomas Beck}
\author[B. Hanin]{Boris Hanin}

\address[T. Beck]{Department of Mathematics, MIT, Cambridge, United States.\medskip}
 \email{tdbeck@mit.edu}
\address[B. Hanin]{Department of Mathematics, Texas A\&M, College
  Station, United States.\medskip}
\email{bhanin@math.tamu.edu}
\begin{document}
\maketitle
\begin{abstract}
We study properties of the nodal sets of high frequency eigenfunctions
and quasimodes
for radial perturbations of the Harmonic Oscillator. In particular, we
consider nodal sets on spheres of large radius (in the classically
forbidden region) for quasimodes with energies lying in intervals
around a fixed energy $E$. For well chosen intervals we show that
these nodal sets exhibit quantitatively different behavior compared to
those of the unperturbed Harmonic Oscillator. These energy intervals
are defined via a careful analysis of the eigenvalue spacings for the
perturbed operator, based on analytic perturbation theory and
linearization formulas for Laguerre polynomials.   
\end{abstract}

\section{Introduction}\label{S:introduction}
In this article we obtain new information about nodal (i.e. zero) sets
of high frequency eigenfunctions 
and eigenvalue spacings for semi-classical Schr\"odinger operators
that are small radial perturbations of the
isotropic Harmonic Oscillator: 
\begin{equation}\label{E:operator-def}
 \Op_\hbar(\ep):= \HO_\hbar + \ep \hbar V(\abs{x}^2),\quad
 \HO_\hbar:=-\frac{\hbar^2}{2}\Delta_{\R^d} +
 \frac{\abs{x}^2}{2},\qquad d\geq 2.
\end{equation}

See \eqref{E:potential-assumptions} for the assumptions we place on $V.$
Although much is known about nodal
sets of eigenfunctions of the Laplacian on a compact manifold,
comparatively little has been proved about nodal sets of eigenfunctions of
Schr\"odinger operators $-\tfrac{\hbar^2}{2}\Delta+V(x)$, even on $\R^d.$
When $V(x)\gives +\infty$ as $\abs{x}\gives
\infty,$ such operators have a discrete spectrum and a complete eigenbasis
for $L^2(\R^d,dx).$ Fixing an energy $E>0,$ and letting $\hbar\gives 0$, any
energy $E$ eigenfunction $\psi_{\hbar, E}$ of $-\tfrac{\hbar^2}{2}\Delta+V(x)$ is rapidly oscillating (with frequency
$\hbar^{-1}$) in the classically allowed region
\[\mathcal A_E:=\setst{x\in \R^d}{V(x)\leq E}\]
and exponentially decaying in the
classically forbidden region 
\[\mathcal F_E:=\setst{x\in \R^d}{V(x)>E}.\]
The nodal set of $\psi_{\hbar, E}$ undergoes a qualitative change
as it crosses from $\mathcal A_E$ to $\mathcal F_E.$ This transition
is illustrated in Figure 1. 

In $\mathcal A_E,$ the eigenfunction $\psi_{\hbar,E}$ behaves much like
an eigenfunction of the Laplacian. For instance, if $V$ is real analytic, then
Jin \cite{jin2017semiclassical} proved that for any bounded open $B\subseteq \mathcal
A_E$ there exists $c,C>0$ such that
\begin{equation}\label{E:allowed-zeros}
c\hbar^{-1}\vol(B)\leq \mathcal H^{d-1}\lr{\set{\psi_{\hbar,E}=0}\cap
  B}\leq C\hbar^{-1}\vol(B).
\end{equation}

\begin{figure}
        \centering

        \begin{subfigure}[b]{0.32\textwidth}
                \includegraphics[width=\textwidth]{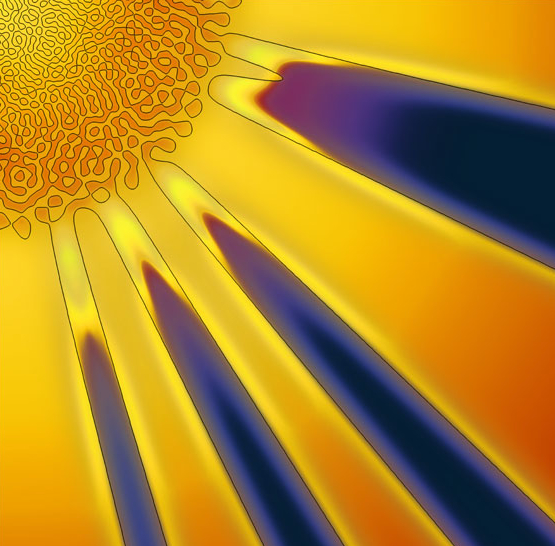}
        \end{subfigure}%
        ~
        \begin{subfigure}[b]{0.32\textwidth}
                \includegraphics[width=\textwidth]{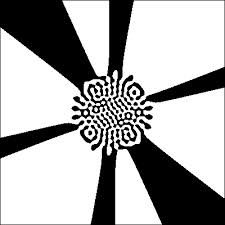}
        \end{subfigure}%
        ~ 
       \begin{subfigure}[b]{0.32\textwidth}
                \includegraphics[width=\textwidth]{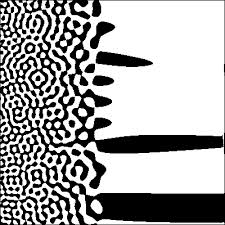}
        \end{subfigure}
\caption{Nodal sets of energy $E$ eigenfunctions of $\HO_\hbar$ have
  qualitatively different behavior in $\mathcal A_E$ and $\mathcal
  F_E.$ On the left, the nodal set is a black line and $\mathcal A_E$
  is a disk with only a quarter shown. This figure was made by Eric
  Heller \cite{Heller-gallery}. In the middle and right, the
  boundary between the white and black regions is the nodal set. The
  allowed region in the middle is the disk and on the right is the
  left half-plane. The figures in middle and on the right are
  reproduced from Bies-Heller \cite{bies2002nodal}.} 
\end{figure}

Here and throughout, $\mathcal H^k$ denotes $k-$dimensional
Hausdorff measure. The same estimates were proved for compact real analytic Riemannian manifolds by Donnelly-Fefferman \cite{donnelly1988nodal} for
eigenfunctions of the Laplacian. Except when $d=1,$ when $\psi_{\hbar,
E}$ has no zeros in $\mathcal F_E,$ much less
is known about the nodal set of $\psi_{\hbar, E}$ in $\mathcal
F_E.$ Jin established that his upper bound in \eqref{E:allowed-zeros}
continues to hold in the forbidden region. 

Aside from this, we are
aware of only several strands of prior work on the subject. The oldest
are the articles of Hoffman-Ostenhof \cite{hoffmann1988asymptotics2, hoffmann1988asymptotics} and
Hoffman-Ostenhoff-Swetina \cite{hoffmann1986continuity,
  hoffmann1987asymptotics} that study nodal for potentials that
vanish at infinity. They show 
that the nodal set of an eigenfunctions on the sphere at infinity
looks locally like the nodal
set of a Hermite polynomial. There is also the paper of Canzani-Toth
\cite{canzani2016nodal} about the persistence of forbidden
hypersurfaces in nodal sets of Schr\"odinger eigenfunctions on a
compact manifold and the articles of B\'erard-Helfer
\cite{berard2014number, berard2015nodal, berard2017some} on nodal
domains for eigenfunctions of the harmonic oscillator and similar
operators (mainly in the allowed region). Finally, we mention the articles of
Hanin-Zelditch-Zhou \cite{hanin2014nodal, hanin2017scaling}, 
which study the typical size of the nodal set in $\mathcal F_E$ and
near the caustic $\dell \mathcal A_E=\set{\abs{x}^2=2E}$ for
\textit{random} fixed energy eigenfunctions of $\HO_\hbar$. We also
refer the reader to the interesting heuristic physics paper of
Bies-Heller \cite{bies2002nodal}. 

In
particular, in \cite{hanin2014nodal} it is shown that for every bounded $B\subseteq
\mathcal F_E$ there exists $C>0$ depending only on the minimum and maximum
distance from a point in $B$ to $\mathcal A_E$ so that
\begin{equation}\label{E:avg-fobidden-zeros}
\mathbb E\left[\mathcal H^{d-1}\lr{\set{\psi_{\hbar, E}=0}\cap B}\right]= C\hbar^{-1/2} \vol(B)\lr{1+O(\hbar)}.
\end{equation}
While the typical nodal density for $\psi_{\hbar, E}$ in $\mathcal F_E$ is therefore $\hbar^{-1/2}$, there are no mathcing deterministic upper and
lower bounds. Indeed, for every bounded open $B\subseteq
\mathcal F_E$
\begin{align*}
\inf_{ \psi_{\hbar, E}\in \ker\lr{\HO_\hbar- E}} \mathcal
H^{d-1}\lr{\set{\psi_{\hbar, E}=0}\cap B}&=0\\
\sup_{  \psi_{\hbar, E}\in \ker\lr{\HO_\hbar- E}} \mathcal
  H^{d-1}\lr{\set{\psi_{\hbar,E}=0}\cap B}&=C\hbar^{-1}\vol(B).
\end{align*}
The infimum is attained when $\psi_{\hbar,E}$ is the unique radial eigenfunction
of $\HO_\hbar$ with given energy $E,$ which has no nodal set whatsoever in $\mathcal F_E,$
and the supremum is attained when $\psi_{\hbar,E}$ is any of the purely
angular eigenfunctions, which are eigenfunctions of the Laplacian on
$S^{d-1}$ of frequency $\approx \hbar^{-1}.$

The difference in the exponents in the various estimates above
raises the question of what happens to the nodal sets of
eigenfunctions for other Schr\"odinger operators. We take up this
question in the present article for the small radial perturbations $\Op_\hbar(\ep)$
\eqref{E:operator-def} of the harmonic oscillator. We are concerned
primarily with the behavior of nodal sets on the sphere at infinity for
eigenfunctions of $\Op_\hbar(\ep)$ with approximately the same energy. Our main results in this direction
are Theorems \ref{T:Zeros} and \ref{T:Zeros2}, which establish upper
and lower bounds on the size of the nodal set of both eigenfunctions
and certain quasi-modes near a fixed energy $E$.

Since $\Op_\hbar(\ep)$ is rotationally symmetric for all $\ep,$ 
its eigenfunctions can be obtained by separating variables (see
\eqref{E:perturbed-spectrum} and \eqref{E:perturbed-eigenfunctions}). The radial
parts of these separation of variables eigenfunctions are deformations
in $\ep$ of the Laguerre functions \eqref{E:Laguerre}, while the
angular parts are the eigenfunctions of the Laplacian on the
round sphere $S^{d-1}.$ At a fixed energy $E$, all such products have
the same rate of growth at infinity when $\ep=0$ (see
\eqref{E:HO-Efns} in \S \ref{S:HO-spec-theory}), and 
hence spherical harmonics of many different angular momenta may
contribute to the nodal set of eigenfunctions at infinity. 

However, for $\ep\neq 0,$ the energies $E_{\ell,n}^V(\ep),$ defined in
\eqref{E:perturbed-spectrum}, for different
angular momenta $\ell$ will no longer
be the same (Theorem \ref{T:Main}). Hence, since the rate of growth at
infinity of the radial eigenfunctions is an increasing function of
$E_{\ell,n}^V(\ep)$ (Proposition \ref{prop:growth}), we see that the nodal
sets at infinity of energy $\approx E$ eigenfunctions and
quasimodes for $\Op_\hbar(\ep)$ depend on the
level spacings of the perturbed energies $E_{\ell,n}^V(\ep)$ for various
angular momenta $\ell$. We obtain precise information on these level
spacings, for what we call
slowly-varying potentials $V$, in Theorem \ref{T:Main}, which is our
main techincal result. 

\section{Statement of Results} 
\noindent Theorem \ref{T:Main} concerns the eigenvalue spacings
for $\Op_\hbar(\ep).$ It holds for $V$ that satisfy 
\begin{equation}\label{E:potential-assumptions}
  V\in C^\infty(\R_+, \R),\quad \limsup_{|x|\gives \infty} |x|^{\eta}V(|x|^2) \leq C,\quad
  V(0)=V'(0)=0
\end{equation}
for some $\eta>0$ and are slowly varying in
the sense of Definition \ref{D:slowly-varying} below. The last
assumption in \eqref{E:potential-assumptions} is only a matter of convenience since $V(0)$
(resp. $V'(0)$) can be absorbed as shifts (resp. scalings) of the 
spectrum of $P_\hbar(0)=\HO_\hbar$. 

Since $\Op_\hbar(\ep)$ is rotationally symmetric for all $\ep$ its
spectrum can be decomposed as a union (with multiplicity):
\begin{equation*}
\Spec\lr{\Op_\hbar(\ep)}= \bigcup_{\ell\geq 0}
\Spec\lr{\Op_{\hbar,\ell}(\ep)},
\end{equation*}
where $\Op_{\hbar, \ell}(\ep):= \Op_{\hbar}(\ep)\big|_{L_\ell^2}$ is
the restriction of $\Op_{\hbar}(\ep)$ to functions with fixed angular 
momentum:  
\[L_\ell^2 = L^2\lr{\R_+,\, r^{d-1}dr}\widehat{\otimes}
\ker\lr{\Delta_{S^{d-1}}+\ell \lr{\ell + d-2}},\qquad L^2(\R^d, dx)=
\bigoplus_{\ell \geq 0} ~L_\ell^2.\]
In the previous line, $\Delta_{S^{d-1}}$ is the Laplacian for the round metric on
$S^{d-1},$ whose spectrum is $\set{-\ell\lr{\ell + d-2}}_{\ell \geq
  0}.$ The spectrum 
\begin{equation}\label{E:perturbed-spectrum}
\Spec\lr{\Op_{\hbar, \ell}(\ep)}=\set{E_{\ell,
    n}^V(\ep)}_{\substack{n\geq \ell,\, n\equiv \ell \text{ (mod
      2)}}},\qquad E_{\ell, n}^V(0)=\hbar\lr{n+\frac{d}{2}}
\end{equation}
of the radial operator for each angular momentum $\ell$ is simple for
small $\ep$ since it is an analytic perturbation 
of the simple spectrum 
\[\Spec\lr{\Op_{\hbar,
    \ell}(0)}=\set{\hbar\lr{n+d/2}}_{\substack{n\geq \ell,\,\,
    n\equiv \ell \text{ (mod 2)}}}.\]     
Let us fix $E>0$ and define
\begin{equation}\label{E:hbar-n}
\hbar_n:=\frac{E}{n+\frac{d}{2}},\qquad n\in \mathbb N.
\end{equation}
At $\ep=0,$ the spectra of the radial oscillators $\Op_{\hbar_n,
  \ell}(0)$ overlap and contain the same energy 
\[E=E_{\ell, n}^V(0)\] 
for all $\ell\leq n$ congruent to $n$ modulo $2.$ However, for small
$\ep$ and generic $V$, we expect the spectra 
of $\Op_{\hbar_n, \ell}(\ep)$ will be disjoint for various $\ell.$
Although we do not have a proof of this fact, Theorem \ref{T:Main}
implies that the eigenspaces of $\Op_{\hbar_n, \ell}(\ep)$ will have
bounded multiplicity uniformly in $n,\ell$ (see \eqref{E:diff-est}). Theorem \ref{T:Main}
concerns the relative positions of the 
perturbations $E_{\ell, n}^V(\ep)$ of $E$ as a function of $\ell.$
 \begin{definition}\label{D:slowly-varying}
Let $E,\delta>0.$ We say a potential $V\in C^\infty(\R_+, \R)$ is
$\delta$-slowly varying in the allowed region for energy $2E$ if it
satisfies \eqref{E:potential-assumptions}, the condition
$\norm{V}_{L^\infty}\leq 1,$ and
\begin{equation}\label{E:slowly-varying}
\frac{\delta^2}{2}\leq \abs{V''(0)}\leq \delta^2\qquad
\text{and}\qquad \sup_{r\in [0,\sqrt{4E} ]}\abs{\frac{V^{(k)}(r)}{k!}}\leq\delta^k
\end{equation}
for all $k\geq 3.$
\end{definition}
\begin{theorem}\label{T:Main}
 There exist constants $C_1,C_2>0$ with the following property. Suppose
   $E>0$, $\delta\in \lr{0, (C_1E)^{-1}}$ and $V$ is $\delta-$slowly varying in
   the allowed region for energy $2E$. Then, for all 
   \begin{equation}\label{E:n-ep-constraints}
n\text{ s.t. }\hbar_n<1,\qquad \ep\in [0, 1/5],\qquad \ell\leq n,\, \ell\equiv n \text{ (mod 2)}
\end{equation}
we have 
 \begin{equation}\label{E:E-expansion}
     E_{\ell, n}^V(\ep) = E +\ep \hbar_nV''(0) \lr{\frac{E}{2}-\frac{d}{4}}^2
     \left[3+\frac{\ell^2}{n^2} \lr{-1+
      S(\ell, n,\ep)} + T(n,\ep) \right] + O\lr{\hbar_n^\infty}.
   \end{equation}
The terms in 
\eqref{E:E-expansion} satisfy the following estimates for every $n,\ell,\ep$ satisfying \eqref{E:n-ep-constraints}:
  \begin{align}
    \label{ST-est}\max \set{\abs{S(\ell,
    n,\ep)}, \abs{T(n,\ep)}} \leq C_2 \cdot \max\set{\delta,\ep}.
  \end{align}

\end{theorem}
\begin{remark}
Here and throughout, a quantity $A_n$ is $O\lr{ \hbar_n^{\infty} }$ if, for each $\gamma\geq1$, there exists a constant $C_{\gamma}$ such that
\begin{align*}
|A_{n}|\leq C_{\gamma}\hbar_n^{\gamma}
\end{align*}
for all $n\geq1$.
\end{remark}


\noindent Theorem \ref{T:Main} shows that $E_{\ell, n}^V(\ep)-E$ is
  essentially a monotone function of $\ell$ if $\delta$ and $\ep$ are
  sufficiently small. More precisely, if $\max\{\delta,\ep\} < (2C_1)^{-1}$ then 
  \begin{equation}\label{E:diff-est}
   \ell'>\frac{\ell}{1-2C_2\max\set{\delta,\ep}} \qquad \Rightarrow \qquad \text{sgn}(V''(0)) \lr{E_{\ell,
        n}^V(\ep) - E_{\ell',n}^V(\ep)}>0.
  \end{equation}

\subsection{Nodal Sets of Eigenfunctions for $\Op_\hbar(\ep)$} In this
section, we state our results on nodal sets. We define
$U_{\ell,n}(\ep)$ to be the span of the eigenfunctions 
of $\Op_{\hbar_n}(\ep)$ of energy
$E_{\ell,n}^V(\ep)$. The vector spaces $U_{\ell,n}(\ep)$ have
multiplicity bounded independent of $n,\ell$ (see \eqref{E:diff-est})). Setting $x\in \mathbb{R}^d\mapsto (r,\omega)$ to be the polar decomposition, $U_{\ell,n}(\ep)$ is spanned by functions of the form 
    \begin{align}
v_{\ell,n,m}(\ep,x) =   \psi_{\ell,n}(\ep,r)Y_{m}^\ell(\omega), \qquad
      1\leq m \leq D_{d,\ell}, \label{E:perturbed-eigenfunctions}
    \end{align}
where $D_{d,\ell}=\dim \ker(\Delta_{S^{d-1}}+\ell(\ell+d-2))$ and the spherical harmonics $Y_m^{\ell}(\omega)$ are an ONB for the
$-\ell(\ell+d-2)$ eigenspace of the Laplacian on $S^{d-1}$:
\begin{equation}\label{E:spherical-harmonics}
\ker(\Delta_{S^{d-1}}+\ell(\ell+d-2))=\text{Span}\set{Y_m^{\ell},\,\, m=1,\ldots, D_{d,\ell}}.
\end{equation}
The function $\psi_{\ell,n}(\ep,r)$ is the unique, tempered,
$L^2(r^{d-1}dr)$-normalized solution to the eigenfunction
equation 
\begin{equation}\label{E:growth1a}
  \Op_{\hbar_n,\ell}(\ep)\psi_{\ell,n}=E_{\ell,n}^V(\ep)\psi_{\ell,n}.
\end{equation}
The energy $E_{\ell,n}^{V}(\ep)$ controls the rate of growth of
$v_{\ell,n,m}(\ep,x)$ for large $|x|$ in the following sense. 
\begin{proposition} \label{prop:growth}
Fix $\eta>0$, and let $V\in C^{\infty}(\mathbb{R}_{+};\mathbb{R})$, such that
\begin{align} \label{E:V-assumption}
\lim\sup_{r\to\infty} r^{\eta}V(r^2) < \infty.
\end{align}
Then, there exists a finite constant $C^V_{\ell,n,\ep}\neq 0$ such that
\begin{align*}
\lim_{r\to\infty}\frac{d^j}{dr^j}\lr{r^{N}e^{\frac{r^2}{2\hbar_n}}\psi_{\ell,n}(\ep,r)}
  = C^V_{\ell,n,\ep}\cdot \delta_{0,j},
\end{align*}
where $N = \frac{1}{\hbar_n}E_{\ell,n}^{V}(\ep) - \frac{d}{2}.$
\end{proposition}
Proposition \ref{prop:growth} is essentially a classical result (see
\S\S 3.1-3.4 in \cite{erdelyi2010asymptotic}). We give a brief derivation in \S
\ref{S:growth}. Our next result concerns the nodal sets of eigenfunctions of
$\Op_{\hbar_n}(\ep)$ whose eigenvalue $E_{\ell,n}^V(\ep)$ nearly
extremizes the distance to $E=E_{\ell,n}^V(0)$, and hence, by Theorem
\ref{T:Main}, are close to
\[E_{0,n}^V(\ep)=E+\lr{\frac{E}{2}-\frac{\hbar_n d}{4}}^2\ep \hbar_n V''(0)(3+O\lr{\max\set{\ep,\delta}})\]
when $\hbar_n$, $\ep$ and $\delta$ are small. Define
\[I_{\ep,\gamma,\hbar_n}=[E_{0,n}^V(\ep) -\hbar_n^{1+2\gamma}, E_{0,n}^V(\ep)+\hbar_n^{1+2\gamma}],\]
for some  $0\leq \gamma \leq 1,$ and consider the span of the corresponding eigenfunctions
\begin{align*}
V_{\ep,\gamma,\hbar_n} =
  \text{span}\{U_{\ell,n}(\ep):E_{\ell,n}^V(\ep)\in I_{\ep, \gamma, \hbar_n}\}.
\end{align*}
\noindent The following concerns the nodal sets of functions in $V_{\ep,\gamma,\hbar_n}$.
\begin{theorem}\label{T:Zeros}
Under the assumptions of Theorem \ref{T:Main}, there exists
$\ep^*,\delta^*, \hbar^*>0$ such that for every $\delta<\delta^*$, $\ep\in[0,\ep^*],$ $\hbar_n<\hbar^*$, and $\gamma\in
[0,1],$ we have
\[\inf_{v\in V_{\ep,\gamma,\hbar_n} }\limsup_{R\gives \infty}\mathcal
  H^{d-2}\lr{\set{v =0}\cap S_R^{d-1}}=0,
\]
and there exist absolute constants $c,C>0$ such that
\[c\hbar_n^{-1+\gamma}\leq \sup_{v\in V_{\ep,\gamma,\hbar_n} }\limsup_{R\gives \infty}\mathcal
  H^{d-2}\lr{\set{v =0}\cap S_R^{d-1}}\leq C
  \hbar_n^{-1+\gamma}.\]
  Here $S_{R}^{d-1}$ is the $d-1$-dimensional sphere of radius $R$
  centred at the origin, and $\mathcal{H}^{d-2}$ is the Haar
  (probability) measure on $S_R^{d-1}$.
\end{theorem}

In the case where the energies $E_{\ell,n}^V(\ep)$ are distinct, we
can conclude additional properties of the nodal sets. 
\begin{theorem}\label{T:Zeros2}
Under the assumptions of Theorem \ref{T:Main}, and the additional assumption that the energies $E_{\ell,n}^V(\ep)$ are distinct, we have the following: For each $v\in V_{\ep,\gamma,\hbar_n}$, 
\[\lim_{R\gives\infty} \mathcal
  H_R^{d-2}\lr{\set{v(x) =0}\cap S_R^{d-1}}\]
  exists, and there exists absolute constants $c, C>0$ such that for
  every $v\in V_{\ep,\gamma,\hbar_n}$ in the complement of a
  co-dimension $1$ subspace, we have
  \begin{align*}
  \lim_{R\gives\infty} \mathcal{H}_R^{d-2}\lr{\set{v(x) =0}\cap S_R^{d-1}} = 0 & \qquad\emph{ if } ~~V''(0) >0 ,  \\ 
c\hbar_n^{-1+\gamma}\leq    \lim_{R\gives\infty} \mathcal{H}_R^{d-2}\lr{\set{v(x) =0}\cap S_R^{d-1}} \leq  C\hbar_n^{-1+\gamma} &\qquad \emph{ if } ~~V''(0) < 0. 
  \end{align*}
 \end{theorem}

\subsection{Acknowledgements} We would like to thank Michael Taylor for
pointing us to his excellent notes \cite{Taylor-dispert} on analytic
perturbation theory. We are also grateful to Vincent Genest for
several useful conversations about special functions and for
directing us to the linearization formulas in \cite{suslov2008hahn}.

\section{Proof of Theorems \ref{T:Zeros} and
  \ref{T:Zeros2}}
We now explain how to derive Theorems \ref{T:Zeros} and \ref{T:Zeros2}
from Theorem \ref{T:Main} and Proposition \ref{prop:growth}. We then
prove Proposition \ref{prop:growth} in \S \ref{S:growth} below and
Theorem \ref{T:Main} in \S \ref{S:main-proof}. Our derivation relies
on several well-known properties of the spherical
harmonics $Y_{m}^{\ell}(\omega)$. The first is an estimate on the measure of the total nodal
set: Since for each $1\leq m \leq D_{d,\ell}$, $Y_{m}^{\ell}(\omega)$
is an eigenfunction on the round sphere $S^{d-1}$ with eigenvalue
$-\ell(\ell+d-2)$, the Donnelly-Fefferman bounds
\cite{donnelly1988nodal} show that there exist constants $c,C>0$ such that
\begin{equation}\label{E:sh-nodal-meas}
c\ell \leq  \mathcal H^{d-2}\lr{\left\{\omega\in
    \mathbb{S}^{d-1}:\sum_{1\leq m\leq D_{d,\ell}}a_{m}Y_m^\ell(\omega)} =0\right\}\leq C\ell.
\end{equation}
In particular, since $E_{0,n}^V(\ep)\in I_{\ep,\gamma,\hbar_{n}}$ by
construction for all $0\leq\gamma\leq 1$, this immediately gives the first statement of Theorem \ref{T:Zeros}. 
Moreover, by a simple application of the Crofton formula ( see
for example Theorem \cite{gichev2009some}), the upper bound in
\eqref{E:sh-nodal-meas} holds also for (non-identically zero) linear
combinations of spherical harmonics up to frequency $\ell:$
\begin{equation}\label{E:sh-nodal-meas2}
 \mathcal H^{d-2}\lr{\left\{\omega\in
    \mathbb{S}^{d-1}:\sum_{s\leq \ell, 1\leq m\leq D_{d,\ell}}a_{s, m} Y_m^s(\omega)=0\right\} }\leq C\ell.
\end{equation}
By the approximate monotonicity of the energies $E_{\ell,n}^{V}(\ep)$
from Theorem \ref{T:Main}, provided $\ep,\delta,\hbar_n$ are sufficiently small, the largest value of $\ell$ for which
$E_{\ell,n}^V(\ep)\in I_{\ep,\gamma,\hbar_n}$, is bounded above and below by a constant multiplied by $\hbar_n^{-1+\gamma}$. The second statement in Theorem \ref{T:Zeros} then follows from the lower bound in \eqref{E:sh-nodal-meas} and the upper bound in \eqref{E:sh-nodal-meas2}. 

To prove Theorem \ref{T:Zeros2} we have the extra assumption that the energies $E_{\ell,n}^V(\ep)$ are distinct. In this case, by Proposition \ref{prop:growth}, the radial part of the eigenfunctions, $\psi_{\ell,n}(\ep,r)$, grows at different rates as $r\to\infty$ for different values of $\ell$. Given $v(x)\in V_{\ep,\gamma,\hbar_n}$, we can write
\begin{align*}
v(x) = \sum_{\ell\in J_{\ep,\gamma,\hbar_n}}\sum_{1\leq m\leq D_{d,\ell}} a_{\ell,m}\psi_{\ell,n}(\ep,r)Y_m^{\ell}(\omega) .
\end{align*}
Among the values of $\ell$ for which $a_{\ell,m}\neq0$ for some $m$, let $\ell^*\in J_{\ep,\gamma,\hbar_n}$ correspond to the largest energy  $E_{\ell,n}^V(\ep)$. Then, by Proposition \ref{prop:growth}, the function $ r^{N^*}e^{\frac{r^2}{2\hbar_n}}v(r,\omega) $ converges in $C^{\infty}(\mathcal{S}^{n-1})$ to 
\begin{align} \label{E:linear-combination}
C_{\ell^*,n,\ep}^V\sum_{1\leq m\leq D_{d,\ell^*}}a_{\ell^*,m}Y_{m}^{\ell}(\omega)
\end{align}
as $r\to\infty$, where $N^*= \frac{1}{\hbar_n}E_{\ell^*,n}^V(\ep) -
\frac{d}{2}$. Moreover, the function in \eqref{E:linear-combination}
has co-dimension $2$ singular set in $\mathbb{S}^{d-1}$ (see
e.g. \cite{hardt1999critical, hardt1987nodal}), and so by
Corollary 2 in \cite{beck2016nodal} we have the convergence of the
nodal set measures, 
\begin{align} \label{E:linear-combination2}
\lim_{R\to\infty}\mathcal
  H_R^{d-2}\lr{\set{v(x) =0}\cap S_R^{d-1}} = \mathcal
  H^{d-2}\lr{\left\{\omega\in\mathbb{S}^{d-1}:\sum_{1\leq m\leq D_{d,\ell^*}}a_{\ell,m}Y_{m}^{\ell}(\omega) =0\right\}} . 
\end{align}
By Theorem \ref{T:Main}, provided $\ep,\delta,\hbar_n$ are sufficiently small, for almost every $v\in V_{\ep,\gamma,\hbar_n}$, $\ell^*$ is equal to $0$ if $V''(0)>0$, while $\ell^* = O(\hbar_n^{-1+\gamma})$ if $V''(0)<0$. Thus, \eqref{E:linear-combination2} together with the Hausdorff measure estimates in \eqref{E:sh-nodal-meas} imply Theorem \ref{T:Zeros2}.

\subsection{Proof of Proposition \ref{prop:growth}} \label{S:growth}
Suppose that $u_{\ell,n}(\ep,r)$ satisfies the equation
\begin{equation} \label{E:growth1a}
\lr{-\frac{\hbar_n^2}{2}\lr{\frac{d}{dr^2} + \frac{d-1}{r}\frac{d}{dr} - \frac{\ell(\ell+d-2)}{r^2}}+\frac{r^2}{2} + \ep\hbar_n V(r^2) - E_{\ell,n}^{V}(\ep)}u_{\ell,n}(\ep,r) = 0
\end{equation}
for $r$ sufficiently large. Setting $t = \frac{r^2}{2\hbar_n}$, and $z_{\ell,n}(\ep,t) = t^{-N/2}e^{t}u_{\ell,n}(\ep,r)$, with $N = \frac{1}{\hbar_n}E_{\ell,n}^V(\ep) - \frac{d}{2}$, this equation becomes
\begin{equation} \label{E:growth2}
\left(\frac{d^2}{dt^2} + 2\lr{-1 + \frac{N/2+d/4}{t}}\frac{d}{dt} + \frac{F_{\ell,n}(\ep,t)}{t^2}\right)z_{\ell,n}(\ep,t) = 0.
\end{equation}
Here the function $F_{\ell,n}(\ep,t)$ is given by
\begin{align*}
F_{\ell,n}(\ep,t) = -\ep tV(2\hbar_{n}t) -\frac{\ell(\ell+d-2)}{4} +\frac{N}{2}(N/2-1) + \frac{Nd}{4},
\end{align*}
and so by the assumption on $V(t)$ from \eqref{E:V-assumption},
\begin{align} \label{E:F-assumption}
\left|F_{\ell,n}(\ep,t)\right| \leq C_{\ell,n,\ep} t^{1-\eta/2},
\end{align}
for a constant $C_{n,\ep}$ for large $t$. Then, for fixed $\ell,n,\ep$, Erdelyi \cite{erdelyi2010asymptotic} (with $\omega = -1$, $\rho = -N/2-d/4$) gives a solution $z^{(1)}_{\ell,n}(\ep,t)$ to \eqref{E:growth2} for $t\geq t_0$, under the assumption on $F_{\ell,n}(\ep,t)$ in \eqref{E:F-assumption}, such that
\begin{align} \label{E:growth3}
\lim\sup_{t\to\infty} t^{\eta}\left|z^{(1)}_{\ell,n}(\ep,t)-1\right| <\infty.
\end{align}
(Note that in \cite{erdelyi2010asymptotic}, equation (7), the assumption placed on $F_{\ell,n}(\ep,t)$ is that it is bounded in $x$, but the same proof works for the sub-linear growth from \eqref{E:F-assumption}.)

Going back to the original function $u_{\ell,n}(\ep,r)$, we obtain a solution $u^{(1)}_{\ell,n}(\ep,r)$ to \eqref{E:growth1a} for $r\geq r_0$, which is non-zero, and satisfies
\begin{align} \label{E:growth4}
\lim_{t\to\infty}t^{N/2}e^{t} u^{(1)}_{\ell,n}(\ep,\sqrt{2\hbar_n t}) = 1,\qquad  \lim_{t\to\infty}\frac{d^j}{dt^j}\left(t^{N/2}e^{t} u^{(1)}_{\ell,n}(\ep,\sqrt{2\hbar_n t})\right) = 0
\end{align}  
for $N = \frac{1}{\hbar_n}E_{\ell,n}^{V}(\ep) - \frac{d}{2}$, $j\geq1$. To obtain another solution to \eqref{E:growth1a} for large $r$, we first set $w^{(1)}_{\ell,n}(\ep,r) = r^{\frac{d-1}{2}} u_{\ell,n}^{(1)}(\ep,r)$, to remove the coefficient of $\frac{d}{dr}$ in \eqref{E:growth1a}. Then defining $u^{(2)}_{\ell,n}(\ep,r) =  r^{-\frac{d-1}{2}}w^{(2)}_{\ell,n}(\ep,r)$, where
\begin{align*}
w^{(2)}_{\ell,n}(\ep,r) = w^{(1)}_{\ell,n}(\ep,r)\int_{r_0}^{r}w^{(1)}_{\ell,n}(\ep,s)^{-2} d s,
\end{align*}
gives the other linearly independent solution to \eqref{E:growth1a}
for $r\geq r_0$. Since $u^{(2)}_{\ell,n}(\ep,r)$ grows exponentially
as $r$ tends to infinity, it is not $L^2(r^{d-1} dr)$-normalisable,
and so our eigenfunction $\psi_{\ell,n}(\ep,r)$ must be proportional
to $u^{(1)}_{\ell,n}(\ep,r)$ for $r\geq r_0$. The proposition then
follows from the estimates in \eqref{E:growth4}.  \qed

\section{Background to Proof of Theorem \ref{T:Main}}

\subsection{Spectral Theory of $\HO_\hbar$}\label{S:HO-spec-theory}
The spectrum of the isotropic harmonic oscillator
$\HO_\hbar$ is 
\[\text{Spec}(\HO_\hbar)=\set{\hbar\lr{n+d/2}}_{n\in \N}.\]
In this article, we will use repeatedly properties of the radial
eigenfunctions of $\HO_\hbar$, which we now recall. Recall from
\eqref{E:spherical-harmonics} the spectrum of the Laplacian on
$S^{d-1}$ and the corresponding real-valued eigenfunctions
$\set{Y_m^{\ell},\,\, m=1,\ldots, D_{d,\ell}}.$ A standard calculation
shows that an ONB for 
$\ker(\HO_\hbar - \hbar\lr{n+d/2})$ is given by 
\[\psi_{\hbar, \ell, n}(r)\cdot Y_m^{\ell}(\w),\qquad 0\leq \ell \leq
n,~~\ell \equiv n~(\text{mod }2),\quad m=1,\ldots, D_{d,\ell},\]
where $x\in \R^d \mapsto (r,\w)$ is the polar decomposition and
\begin{align} \label{E:Laguerre}
\psi_{\hbar, \ell, n}(r)= \hbar^{-\frac{\ell}{2}-\frac{d}{4}}\mathcal N_{\ell,n}\cdot
r^{\ell}e^{-r^2/2\hbar} L_{n'}^{\lr{\alpha}}\lr{r^2/\hbar},\qquad \mathcal
N_{\ell, n}^{\,2} = \frac{2\cdot\Gamma\lr{\frac{n-\ell}{2}+1}}{\Gamma\lr{\frac{n+\ell +d}{2}}}.
\end{align}
In the above, we have set
\[n':=\frac{n-\ell}{2},\qquad \alpha:=\ell + \frac{d-2}{2},\]
and denoted by $L_k^{(\alpha)}$ the generalized Laguerre
polynomials. We often fix $E>0$ and define $\hbar=\hbar_n$ to be a
function of $n$ and $E$ as in \eqref{E:hbar-n}. In this case, we
abbreviate 
\[\psi_{\ell, n}:=\psi_{\hbar_n, \ell, n}.\]
As explained in the introduction, the energy $E$ determines a
clasically forbidden region $\mathcal F_E=\set{r^2>2E}$, where the
fixed energy eigenfunctions $\psi_{\hbar_{n},\ell, m}$ for $m\approx n$ are 
uniformly exponentially small. More precisely, for any $\ep>0$ there exists
$C>0$ such that 
\begin{equation}\label{E:Agmon}
  \sup_{\substack{0\leq \ell \leq m,\,\, \abs{m-n}<\frac{n}{2}\\ \ell \equiv m \text{ (mod 2)}\\
      r\in [\sqrt{4E},\infty)}}\abs{e^{(1-\ep)r^2/2}\psi_{\hbar_n,\ell,
      m}(r)} \leq C.
\end{equation}
Since $\lim_{x\to\infty}x^{-k}L_k^{(\alpha)}(x)=\frac{(-1)^k}{k!}$,  for each
fixed $n,$ the radial eigenfunctions $\psi_{\ell,n}$ differ at infinity
only by a constant: 
\begin{equation}\label{E:HO-Efns}
0 < \lim_{r\gives \infty} \frac{\psi_{\ell, n}(r)}{\hbar_n^{-\frac{n}{2}-\frac{d}{4}}r^n e^{-r^2/2\hbar}}  =
\mathcal N_{\ell, n}<
\infty\qquad \forall 0\leq \ell\leq n,\quad \ell \equiv n ~(\text{mod }2).
\end{equation}

\subsection{Linearization Formulas for Laguerre Functions} In order to
perform perturbative calculations about $\Spec\lr{\Op_\hbar(\ep)},$ we
will need a convenient expression for 
\[J_{a,b,k}^\alpha:=\int_0^\infty e^{-\rho}\rho^{\alpha+k}
L_a^{\lr{\alpha}}(\rho)L_b^{\lr{\alpha}}(\rho)d\rho,\]
where as in \eqref{E:Laguerre} $L_a^{\lr{\alpha}}$, $L_b^{\lr{\alpha}}$ are the
generalized Laguerre polynomials. 
\begin{proposition}[Special Case of \cite{suslov2008hahn} Eqn. (2.5)]
For any $a,b,k\in \mathbb N$ and every $\alpha>-1,$ we have
\begin{equation}\label{E:overlap3F2}
J_{a,b,k}^\alpha =  \lr{\alpha+1}_k\binom{k}{\abs{a-b}}\left[\frac{(a\lor
                 b + 1)_{\abs{a-b}}}{\lr{a\lor b + \alpha +
                 1}}_{\abs{a-b}}\right]^{1/2} \pFq{3}{2}{-k,
                 \,k+1,\,-(a\lor b)}{\abs{a-b}+1, \alpha +1}{1}.
\end{equation}
\end{proposition}
\noindent We have written $a\lor b$ for the minimum of $a,b$, 
\[\lr{x}_q:=\frac{\Gamma(x+q)}{\Gamma(x)}\]
for the Pochhammer symbol and 
\[ \pFq{3}{2}{a_1,
                 \,a_1,\,a_3}{b_1, b_2}{1} =
               \sum_{q=0}^\infty\frac{1}{q!}
               \frac{\lr{a_1}_q\lr{a_2}_q\lr{a_3}_q}{\lr{b_1}_q\lr{b_2}_q}\]
for a hypergeometric function. In the case of \eqref{E:overlap3F2}, the
sum terminates at $q=k\lor a\lor b$ since $\lr{x}_q$ vanishes for
$x=-1,-2,\ldots.$ The expression in \eqref{E:overlap3F2} differs
slightly from the one in \cite{suslov2008hahn} because our Laguerre functions
are $L^2-$normalized while the ones in \cite{suslov2008hahn} are
not. 

\subsection{Analytic Perturbation Theory}\label{S:perturbation-theory} We recall in this section
several results from analytic perturbation theory. These results are
 classical, and we mainly follow the notes
\cite{Taylor-dispert} of M. Taylor. Suppose that $H$ is an unbounded
self-adjoint operator on a Hilbert space $\mathcal H$ with discrete
spectrum $\set{\leb_j}_{j=0}^\infty$ and corresponding eigenfunctions
\[\mathcal H u_j = \leb_j u_j.\]
Suppose further that $W$ is a bounded self-adjoint operator on $\mathcal
H.$ Consider some $\leb =\leb_n\in \Spec\lr{H},$ and write $u = u_n$ for the
corresponding eigenfunction. Then for 
all $\ep$ sufficiently small the operator 
\[H(\ep):=H+\ep W\] 
has a simple eigenvalue $\leb(\ep)$ with 
\[\lr{H+\ep W}u(\ep)=\leb(\ep)u(\ep).\]
Both $\leb(\ep)$ and $u(\ep)$ are analytic in $\ep.$ Explicitly, write
\[\leb(\ep)= \leb + \ep \sum_{k\geq 0} \ep^k \mu_k,\quad u(\ep)=
u + \ep \sum_{k\geq 0} \ep^k v_k,\]
and impose the normalization
\[\lr{u(\ep)-u}\perp u.\]
We have the following recursive formulas for $\mu_k, v_k$ for each
$k\geq 0$
\begin{equation}
  \begin{cases}
    w_k=\sum_{m=0}^{k-1} \mu_{k-m-1}v_m, & ~~ w_0=0\\
    v_k = \lr{H-\leb}^{-1}\left[\Pi_u^\perp \lr{Wv_{k-1}} +
      w_k\right],&~~v_{-1}=u\\
    \mu_k = \inprod{Wv_{k-1}}{u}
  \end{cases}.
\end{equation}
The operator $\Pi_u^\perp$ is the projection onto the orthogonal
complement of $u.$ Using this recursion and integration by parts, we
have for any $X\in \mathcal H$ 
\begin{equation}\label{E:G-def}
\inprod{Wv_k}{X}= - \inprod{Wv_{k-1}}{G\lr{X}}, \qquad G=\lr{H-\leb}^{-1}\circ \Pi_u^\perp \circ W.
\end{equation}
Writing $u=u_n$ and using \eqref{E:G-def}, we find for $k\geq 0$
\begin{equation}\label{E:muj-G}
  \mu_k = (-1)^k\inprod{Wu}{G^{(k)}(u)}.
\end{equation}
Using the definition of $G$ we obtain
\begin{align}
\notag  \mu_0 &= \inprod{Wu}{u}\\
\mu_k &= \sum_{m_1,\ldots, m_k\neq
      n} \inprod{W u }{u_{m_1}} \prod_{i=1}^k
    \frac{\inprod{Wu_{m_i}}{u_{m_{i+1}}}}{\leb_n-\leb_{m_i}},\qquad k \geq 1,
\end{align}
with the convention that $u_{m_{k+1}}=u.$ We will also need the following simple estimate. 
\begin{Lem}\label{L:level-spacing}
Suppose that $H$ not only has simple spectrum but also that the spacing
between any two consecutive eigenvalues is bounded below by $\eta>0.$
Then, if $\norm{W}_{L^\infty}\leq \eta,$ 
 \begin{equation}
  \sup_{\substack{n\in \mathbb N\\\ep\in [0,\frac{1}{5}]}} \abs{\leb_n(\ep)- \leb_n(0)}< \eta/4.
 \end{equation}
\end{Lem}
\begin{proof}
Let $\mu_k$ given by \eqref{E:muj-G}. Since
\[\norm{G}_{\mathcal H\gives \mathcal H}\leq \eta^{-1}\norm{W}_{L^\infty}\leq 1,\]
we conclude 
\[\abs{\mu_k}\leq 1.\]
Thus, for $\ep\in[0,\tfrac{1}{5}]$ we can write 
\[\leb_n(\ep)-\leb_n(0) = \ep \sum_{k=0}^\infty \ep^k \mu_k,\]
and, in particular
\[\abs{\leb_n(\ep)-\leb_n(0)} \leq \eta\cdot \frac{\ep}{1-\ep}\leq \frac{\eta}{4}.\]
\end{proof}

\section{Proof of Theorem \ref{T:Main}}\label{S:main-proof}
\noindent Throughout this section, we fix $E>0$ and use the convention
\[\hbar = \hbar_n = \frac{E}{n+\frac{d}{2}}\]
as in \S \ref{S:introduction}. The proof of Theorem \ref{T:Main} 
consists of three steps, which we describe below.   
\subsection{Step 1.} The first step is to replace both $V(r^2)$ and
$E_{\ell, n}^V(\ep)$ by an 
$\hbar-$dependent Taylor series around $r=0$ and $\ep=0,$
respectively. More precisely, for each $K\in \mathbb N,$ define
\[ V_K(r):=\sum_{k=0}^K \frac{V^{(k)}(0)}{k!} r^{2k}.\]
\begin{proposition}\label{P:step1}
There exists a constant $C_1>0$ with the following property. For
  all $E>0,$ any $\delta \in (0,(C_1 E)^{-1})$, each $\delta-$slowly varying
  potential $V,$ and every $n,\ep$ such that  $\hbar_n<1$ and 
  $\ep\in [0,1/5]$, we have 
  \begin{equation}\label{E:eval-jets}
\sup_{
    \ell\leq n,\, \ell\equiv n \text{ (mod 2)}}\abs{E_{\ell, n}^V(\ep) - \sum_{j=0}^J \frac{\lr{\hbar_n\ep}^j}{j!}
\frac{d^j}{d\ep^j}\bigg|_{\ep=0} E_{\ell,n}^{V_K}(\ep)} = O\lr{ \hbar_n^{\infty} }
\end{equation}
provided $K=K(n)$ and $J=J(n)$ satisfy
\begin{equation}\label{E:KJ-assumptions}
\limsup_{n\gives \infty} \frac{K(n)}{\log n}=\limsup_{n\gives
  \infty} \frac{J(n)}{\log n}=\infty\qquad \text{and}\qquad
\lim_{n\gives \infty} \frac{K(n)J(n)}{n}=0. 
\end{equation}
\end{proposition}

The approximation \eqref{E:eval-jets} is the source of
  the $O(\hbar_n^\infty)$ error in \eqref{E:E-expansion}. The
function $E_{\ell,n}^{V_K}(\ep)$ whose jets appear in Proposition 
\ref{P:step1} is formally defined in the same way as
$E_{\ell,n}^V(\ep).$ However, note that $V_K$ is not a bounded
operator on $L^2([0,\infty), r^{d-1}dr).$ It therefore does not
strictly follow from the discussion in \S \ref{S:perturbation-theory}
that these jets are well-defined. Nonetheless, we simply
\textit{define} these jets by $\mu_{\ell,n}^{V_K}(0)=E$ and for $j\geq
1$
\begin{equation}\label{E:jets-def}
\mu^{V_K}_{\ell, n}(j):= \frac{1}{j!}\frac{d^j}{d\ep^j}\bigg|_{\ep=0} E_{\ell,n}^{V_K}(\ep) = \lr{-1}^{j-1} \inprod{V_K
  \psi_{\hbar_n,\ell,n}}{G_{\ell,K}^{(j-1)} (\psi_{\hbar_n, \ell,n})},
\end{equation}
where the inner product is in $L^2([0,\infty), r^{d-1}dr)$ and 
\[G_{\ell,K}:=\lr{\Op_{\hbar} - E^V_{\ell,n}(0)}^{-1}\circ \Pi_{\psi_{\hbar_n,
     \ell,n}^{\perp}}\circ V_K.\]
The inner products on the right hand side of \eqref{E:jets-def} are
finite provided $K(n)$ satisfies \eqref{E:KJ-assumptions} by the Agmon
estimates \eqref{E:Agmon}. We prove Proposition \ref{P:step1} in \S
\ref{S:step1-proof}. 

\subsection{Step 2} The second step in the proof of Theorem
\ref{T:Main} is to write the derivatives of $E^{V_K}_{\ell,n}(\ep)$ at
$\ep=0$ that appear in Proposition \ref{P:step1} in 
terms of hypergeometric functions and obtain their
asymptotics. Unwinding the definition of $G_K,$ using
\eqref{E:jets-def}, and recalling that the spectrum of
$\HO_\hbar=\Op_\hbar(0)$ has level spacings $\hbar,$ we may
write $\mu^{V_K}_{\ell, n}(j)$ as  
\begin{equation}\label{E:jth-var-sumproduct}
   \hbar_n^{1-j}\sum_{\substack{m_1,\ldots, m_{j-1}\neq
      n\\ \abs{m_i - m_{i+1}}\leq 2K\\ \abs{n-m_1}\leq 2K}} \inprod{V_K\psi_{\hbar_n,
      \ell,n}}{\psi_{\hbar_n, \ell,m_1}} \prod_{i=1}^{j-1}
    \frac{\inprod{V_K \psi_{\hbar_n,\ell,m_i}}{\psi_{\hbar_n, \ell,m_{i+1}}}}{m_i-n}
\end{equation}
with the convention that $m_{i+1}=n.$ The restriction that
$\abs{m_i-m_{i+1}}\leq 2K$ comes from the binomial 
coefficient in \eqref{E:overlap2} below. 

To state our next result, we
augment the notation in \S \ref{S:HO-spec-theory} and
 write for each $n\geq \ell,\,\, \ell \equiv n \text{ (mod 2)}$ and all
$s,t\geq \ell$ with $s,t\equiv \ell \text{ (mod 2)}$
\begin{equation}\label{E:prime-defs}
s': = \frac{s-\ell}{2},\quad t': = \frac{t-\ell}{2},\quad \alpha :=
\ell + \frac{d-2}{2},\quad s'\lor t' : = \min\set{s',t'}.
\end{equation}
For $n\in\mathbb{N}$, we will be interested in the values of $s$, $t$, and $\ell$ in the set
\begin{equation} \label{E:parameter-space}
U^n_{s,t,\ell} = \left\{(s,t,\ell) \in \mathbb{N}^3:  s,t,\ell\equiv n \text{ (mod 2)}, \ell \leq n, |s-n|<\frac{n}{2}, |t-n|<\frac{n}{2},\right\}. 
\end{equation}
For each
$K,s,t\in \mathbb N,$ we recall our assumptions $V(0)=V'(0)=0$ and write
\begin{equation}\label{E:overlap1}
  \inprod{V_K\psi_{\hbar_n, \ell,s}}{\psi_{\hbar_n, \ell,t}} =
  \sum_{k=2}^K \frac{V^{(k)}(0)}{k!} \hbar_n^k A_{k,s,t,\ell}.
\end{equation}
The following Proposition is proved in \S \ref{S:step2-proof}. 
\begin{proposition}\label{P:step2}
There exist constants $C_1,C_2>0$ with the following
property. For any $E>0$, if $\delta \in
(0, (C_1E)^{-1})$ and $V$ is a $\delta$-slowly
varying potential in the allowed region for energy $2E$ (Definition \ref{D:slowly-varying}), then for each 
\[k_0\geq 2,\qquad n\text{ s.t. } \hbar_{n}<1,\qquad(s,t,\ell)\in
U^n_{s,t,\ell},\]
we have 
\begin{equation} \label{E:overlap2a}
\abs{\sum_{k\geq k_0}
\frac{V^{(k)}(0)}{k!} \hbar_n^k A_{k,s,t,\ell} - T(n,s,t)}\leq C_2
\frac{1+\ell^2}{\lr{s\lor t}^2}e^{-|s-t|}\lr{E\delta}^{k_0}.
\end{equation}
Here, as usual $K=K(n)$ satisfies \eqref{E:KJ-assumptions}, and $T(n,s,t)$ is $\ell-$independent and satisfies
\[\sup_{\substack{\hbar_n<\hbar^*,\\(s,t,\ell)\in U^n_{s,t,\ell}}}
e^{|s-t|}\abs{T(n,s,t)}\leq \lr{
C_2E\delta}^{k_0}.\]
  \end{proposition}
  \begin{remark}
    We will only use Proposition \ref{P:step2} for $k_0=2,3.$ Also, we
    will obtain the following exact formula for $A_{k,s,t,\ell}:$ 
\begin{equation}
\label{E:overlap2} A_{k,s,t,\ell}  =
                 \lr{\alpha+1}_k\binom{k}{\abs{s'-t'}}\left[\frac{(s'\lor
                 t' + 1)_{\abs{s'-t'}}}{\lr{s'\lor t' + \alpha +
                 1}}_{\abs{s'-t'}}\right]^{1/2} \pFq{3}{2}{-k,
                 \,k+1,\,-(s'\lor t')}{\abs{s'-t'}+1, \alpha +1}{1},
\end{equation}
  where the notation is from \eqref{E:prime-defs}. In particular, $A_{k,s,t,\ell}=0$ whenever $|s'-t'|>k$. 
  \end{remark}

\subsection{Step 3} 
 The final step in the proof of Theorem \ref{T:Main}
is to observe that combining Proposition 
\ref{P:step2} with the expression for $\frac{1}{j!}\frac{d^j}{d\ep^j}\bigg|_{\ep=0} E_{\ell,n}^{V_K}(\ep)$ from \eqref{E:jth-var-sumproduct} and \eqref{E:overlap1}, we obtain the following estimates. 
\begin{proposition}\label{P:step3} There exist constants $C_1,C_2>0$ with the following
property. Fix $E>0$ and $\delta\in (0,\lr{C_1E}^{-1})$ and a
$\delta$-slowly varying potential $V$. For $n\in \mathbb N,$
 and $K=K(n)$ and $J=J(n)$ satisfying \eqref{E:KJ-assumptions}, we have
\[\sum_{j=2}^J \frac{\lr{\hbar_n \ep}^j}{j!} \frac{d^j}{d\ep^j}\bigg|_{\ep=0} E_{ \ell, n}^{V_K}(\ep)=T_{K}(n,\ep)+\frac{\ell^2}{n^2}S_{K}(\ell, n, \ep).\] 
Here for every $\ell, n,\ep$ satisfying
\begin{equation}\label{E:lnep-constraints}
0\leq \ell\leq n,\,\,\ell \equiv n\text{ (mod 2)},\qquad n\text{
  s.t. }\hbar_n <1,\qquad \ep\in[0,1/5],
\end{equation}
we have the estimates
\begin{align*}
  \max\set{\abs{T_{K}(n,\ep)}, \abs{S_{K}(\ell,n,\ep)}}\leq  C_2
  \hbar_n \ep^2 \lr{E\delta}^2.
\end{align*}
Moreover, in the notation of Proposition \ref{P:step2}, we have
\begin{equation}
  \frac{d}{d\ep}\bigg|_{\ep=0} E_{ \ell,
    n}^{V_K}(\ep)=\inprod{V_K\psi_{\hbar_n, \ell, n}}{\psi_{\hbar_n, \ell,
      n}} = \lr{\hbar_n^2 \frac{V''(0)}{2}A_{2,n,n,\ell} + Y(n)+ \frac{\ell^2}{n^2}X(\ell,n)},
\end{equation}
with
\begin{align*}
  \max\set{\abs{X(\ell,n)},\abs{Y(n)}}\leq C_2\lr{E\delta }^3
\end{align*}
 for $\ell, n,\ep$ satisfying \eqref{E:lnep-constraints}.
\end{proposition}

The proof of Theorem \ref{T:Main} is complete once we choose $C_1$
to be the maximum of the $C_1$'s that are provided by Propositions
\ref{P:step1} and \ref{P:step3}, use that
\[\hbar^2 A_{2,n,n,\ell} = 6\lr{\frac{\hbar_n n }{2}}^2\lr{1-
  \frac{1}{3}\cdot \frac{\ell^2}{n^2} + \frac{2-d}{3n}\cdot
  \frac{\ell}{n} + \frac{d}{n}+\frac{d(d+2)}{6n^2}},\]
and substitute the estimates from Proposition \ref{P:step3} into
\eqref{E:eval-jets}.

\section{Proof of Proposition \ref{P:step2}}\label{S:step2-proof} 
Let us first derive \eqref{E:overlap1} and \eqref{E:overlap2}. Recall from
\eqref{E:Laguerre} that, as a function of the radial
variable $r=\abs{x},$ the radial eigenfunctions of the unperturbed operator ($\ep=0$) are 
\[\psi_{\hbar_n, \ell,
  s}(r)= \hbar_n^{-\frac{\ell}{2}-\frac{d}{4}}\mathcal N_{s,\ell, d}\cdot
r^{\ell}e^{-r^2/2\hbar} L_{\frac{1}{2}\lr{s-\ell}}^{\lr{\ell +
    \frac{d-2}{2}}}\lr{r^2/\hbar_n},\qquad \mathcal N_{s, \ell, d}^{\,2} = \frac{2\cdot\Gamma\lr{\frac{s-\ell}{2}+1}}{\Gamma\lr{\frac{s+\ell +d}{2}}},\]
where $L_k^{(\alpha)}$ are the generalized Laguerre polynomials. Hence, for $\alpha = \ell + \frac{d-2}{2}$, 
we have
\begin{align*}
\inprod{V_K \psi_{\hbar_n, \ell, s}}{\psi_{\hbar_n, \ell, t}} &=\frac{\mathcal N_{s,\ell,d} \mathcal
  N_{t,\ell,d}}{2} \int_0^\infty V_K(\sqrt{\hbar_n \rho}) \rho^{\alpha} e^{-\rho}
                                                               L_{s'}^{\alpha}(\rho)L_{t'}^{\alpha}(\rho)d\rho \\
&=  \frac{\mathcal N_{s,\ell,d} \mathcal N_{t,\ell,d}}{2}
  \sum_{k=0}^K \frac{\hbar_n^kV^{(k)}(0)}{k!} \int_0^\infty \rho^{\alpha+k} e^{-\rho}
                                                               L_{s'}^{\alpha}(\rho)L_{t'}^{\alpha}(\rho)d\rho .
\end{align*}
Writing
\[ \frac{\mathcal N_{s,\ell,d} \mathcal
  N_{t,\ell,d}}{2}=\left[\frac{\lr{(s'\lor
      t')+1}_{\abs{s'-t'}}}{\lr{(s'\lor
      t')+\alpha+1}_{\abs{s'-t'}}}\right]^{1/2}\cdot
\frac{\Gamma\lr{(s'\lor t') +1}}{\Gamma\lr{(s'\lor t')+\alpha +1}}\]
and using equation (2.5) in \cite{suslov2008hahn} then proves
\eqref{E:overlap1} and \eqref{E:overlap2}. Next, we will show that for all $n\in\mathbb{N}$ and $(s,t,\ell)\in U^n_{s,t,\ell}$, $A_{k,s,t,\ell}$ has the expansion
\begin{equation}\label{E:overlap3}
\hbar_n^k A_{k,s,t,\ell} = \lr{\frac{\hbar_n \lr{s\lor t}}{2}}^k 
\left[T_1(k,s,t) + \frac{\ell^2}{\lr{s\lor t}^2}
  T_2(k,s,t,\ell)\right].
\end{equation}
Here for some $C_1>0$, we have
\begin{align} 
\sup_{\substack{s,t\in \mathbb N\\ \abs{s-n},\abs{t-n}\leq
  \frac{n}{2}}}\abs{T_1(k,s,t)}& \leq C \cdot C_1^k,\qquad \sup_{\substack{s,t,\ell\in \mathbb N\\ \abs{s-n},\abs{t-n}\leq
  \frac{n}{2}\\ \ell \leq n,\, \ell \equiv n \text{ (mod 2)}}}
  \abs{T_2(k,s,t,\ell)}\leq C\cdot C_1^k.\label{E:overlap4}
\end{align}
Note that for $s\lor t \leq 3n/2$, we have that $\frac{\hbar_n(s\lor
  t)}{2}\leq  E$. Moreover,  by \eqref{E:overlap2}, $A_{k,s,t,\ell}$
is equal to zero when $|s'-t'| = \tfrac{|s-t|}{2}>k$. Hence, the term
$e^{-\abs{s-t}}$ appearing in \eqref{E:overlap2a} is bounded by
$e^{-2k}$ and can be absorbed into the constant $C_1$ in
\eqref{E:overlap3} and \eqref{E:overlap4}. Thus, these estimates, together with Definition \ref{D:slowly-varying} of a $\delta-$slowly varying
potential allow us to sum over $k$ to establish \eqref{E:overlap2a} and complete the proof of Proposition \ref{P:step2}. To obtain the estimates in \eqref{E:overlap3} and \eqref{E:overlap4}, we need two lemmas, in which we abbreviate
\[N = s\lor t,\qquad \beta = |s'-t'|.\]
In particular, this means that $\abs{N-n}\leq \frac{n}{2}$. Since $A_{k,s,t,\ell} = 0$ for $|s'-t'|>k$,  we can and will restrict to the case where $0\leq\beta \leq k\leq K(n)\ll n.$ 

\begin{Lem}\label{L:prefactor-expansion}
There exists $C_2>0$ such that for every $0\leq \beta\leq k$, $0\leq \ell\leq N$, $\ell\equiv N$ (mod 2), 
\[\abs{\frac{\lr{\frac{N-\ell}{2}+1}_\beta}{\lr{\frac{N-\ell}{2}+\alpha +1
  }_\beta} -\left[ 1 - 2\beta\cdot \frac{\ell}{N}+S(\beta,N)\right]}
\leq  \frac{1+\ell^2}{N^2}\cdot C_2^\beta,\]
where $S(\beta,N)$ is $\ell-$independent and satisfies
\[\abs{S(\beta,N)}\leq \frac{C_2\beta}{N}.\]
\end{Lem}

\begin{Lem}\label{L:3F2-expansion}
There exists $C_3>0$ such that for every $0\leq \beta\leq k$ and each $0\leq
 \ell\leq N,$ $\ell\equiv N$ (mod 2), we have
\begin{align*} & \abs{ (\alpha +1)_k\pFq{3}{2}{-k, \,k+1,\,-N'
}{\beta+1, \alpha +1}{1}- \left[ \frac{(2k)!}{k!(\beta+1)_k}
\lr{\frac{N}{2}}^k\lr{1+ \beta \cdot \frac{\ell}{N} +T(\beta,k,N)}\right]}\\
& \leq
\frac{k(1+\ell^2)}{N^2} C_3^k
\end{align*}
where $T(\beta,k,N)$ is $\ell$-independent and satisfies
\[\sup_{0\leq \beta\leq k}\abs{T(\beta,k,N)}\leq C_3\cdot\frac{k^2}{N}.\]
\end{Lem}

We will prove these lemmas in \S\S \ref{S:lemma10-proof}-\ref{S:lemma11-proof} below. Assuming
them for the moment, we prove \eqref{E:overlap3} and \eqref{E:overlap4} (which were used to complete the proof of Proposition
\ref{P:step2}). Using Lemmas \ref{L:prefactor-expansion} and \ref{L:3F2-expansion}, and the expansion for $A_{k,s,t,\ell}$ from \eqref{E:overlap2} we find that 
we have
  \begin{align*}
&A_{k,s,t,\ell} = \binom{k}{\beta} \lr{\frac{\lr{ N'+
      1}_{\beta}} {\lr{N' + \alpha + 1}_{\beta}}}^{1/2}(\alpha+1)_k\pFq{3}{2}{-k, \,k+1,\,-N'
}{\beta+1, \alpha +1}{1}\\
= &\binom{k}{\beta}\cdot\lr{\frac{N}{2}}^k\cdot\lr{\frac{(2k)!}{k!(\beta+1)_k}\cdot
\bigg[\lr{1+ \beta \cdot \frac{\ell}{N}} + T(\beta,k,N)\bigg] +
O\lr{\frac{1+\ell^2}{N^2} C_2^k}} \\
&\cdot \lr{1 - 2\beta\cdot \frac{\ell}{N} + S(\beta,N) +
   O\lr{\frac{1+\ell^2}{N^2}\cdot C_3^\beta}}^{1/2}\\
&=\binom{k}{\beta}\cdot
\lr{\frac{N}{2}}^k\lr{\frac{(2k)!}{k!(\beta+1)_k}\left[1+ T(\beta,k,N)\right]\cdot[1+S(\beta,N)]^{1/2}+O\lr{\frac{1+\ell^2}{N^2}\cdot (C_3C_2)^k}},
  \end{align*} 
with $S(\beta,N) = O(\tfrac{\beta}{N})$, $T(\beta,k,N)=O\lr{\tfrac{k^2}{N}}.$ Since
\begin{align*}
\binom{k}{\beta}\cdot\frac{(2k)!}{k!(\beta+1)_k}=\binom{k}{\beta}\cdot \binom{2k}{k}\cdot \frac{k!}{(\beta+1)_k} = O(8^k),
\end{align*}
and $A_{k,s,t,\ell} = 0$ unless $|s'-t'|\leq k$, we obtain \eqref{E:overlap3} and \eqref{E:overlap4}. \qed

%

\subsection{Proof of Lemma \ref{L:prefactor-expansion}}\label{S:lemma10-proof}
  We want to estimate
\[h_\beta\lr{\tfrac{\ell}{N}}:=\frac{\lr{\frac{N-\ell}{2}+1}_\beta}{\lr{\frac{N-\ell}{2}+\alpha +1
  }_\beta}=\frac{g_\beta\lr{\tfrac{\ell}{N}}}{f_\beta\lr{\tfrac{\ell}{N}}},\]
where
\[g_\beta(x):=\prod_{j=1}^\beta\lr{1-x+ \frac{2j}{N}},\qquad \text{and}\qquad
f_\beta(x):=\prod_{j=0}^{\beta-1}\lr{1+x + \frac{d + 2j}{N}}.\]
We have the estimates for $|x|\leq 1$, 
\begin{align*}
  g_\beta(x)&= \prod_{j=1}^\beta\lr{1-x+\frac{2j}{N}} \leq
          \lr{2+\frac{2\beta}{N}}^\beta=O\lr{2^{\beta}e^{\beta^2/N}}\\
g_\beta'(x)&=\sum_{j_1=1}^\beta \prod_{j\neq j_1}\lr{1-x+\frac{2j}{N}}= O\lr{2^{\beta}\beta\cdot
         e^{\beta^2/N}}\\
g_\beta''(x)&=\sum_{1\leq j_1<j_2\leq \beta} \prod_{j\neq j_1,\,j_2}\lr{1-x+\frac{2j}{N}}= O\lr{2^{\beta}\beta^2\cdot
         e^{\beta^2/N}},
         \end{align*}
         and we have the analogous estimates for the function $f_{\beta}(x)$. By Taylor's Theorem, 
\[h_\beta\lr{\tfrac{\ell}{N}}= h_\beta(0)+
  \frac{\ell}{N} h'_\beta(0)+\frac{\ell^2}{N^2}O\lr{\norm{h_\beta''}_{L^\infty([0,\ell/N])}},\]
 and by the estimates above, $h_{\beta}(0) = 1 + O(\beta/N)$ and $h_{\beta}'(0) = -2\beta + O(\beta/N)$. 
Since 
\begin{align*}
  h_\beta''(x)= \frac{f_\beta(x)^2\left[g_\beta''(x)f_\beta(x)-f_\beta''(x)g_\beta(x)\right] -
  2f_\beta(x)f_\beta'(x)\left[g_\beta'(x)f_\beta(x)-f_\beta'(x)g_\beta(x)\right]}{f_\beta(x)^4},
\end{align*}
and $f_\beta(x)\geq1$, the estimates above also imply that for any $C_2>8$,
\[\sup_{x\in [0,\ell/N]}\abs{h_\beta''(x)} \leq C_2^\beta.\]
as required. \qed

\subsection{Proof of Lemma \ref{L:3F2-expansion}}\label{S:lemma11-proof}
  By definition,
\[(\alpha +1)_k\pFq{3}{2}{-k, \,k+1,\,-N'
}{\beta+1, \alpha +1}{1} = \sum_{q=0}^k \underbrace{\frac{1}{q!} 
\frac{(-k)_q(k+1)_q}{(\beta+1)_q}}_{=:a_{q,\beta}} \cdot \underbrace{\frac{(-N')_q
  (\alpha+1)_k}{(\alpha+1)_q}}_{=:b_q}.\]
Let us check that there exists $C>0$ so that
\begin{equation}\label{E:small-terms}
\sum_{q=0}^{k-2} a_{q,\beta} b_q =  O\lr{ \frac{1+\ell^2}{N^2}\cdot
  \lr{CN}^k},
\end{equation}
where the implied constant is independent of $\beta,\ell,k,N.$ Note
that $\abs{a_{q,\beta}}\leq \abs{a_{q,0}}.$ Hence, it is sufficient to
establish \eqref{E:small-terms} for $\beta=0.$ Define
\[f(q):=\abs{\frac{a_{q+1,0}}{a_{q,0}}}=\frac{(k-q)(k+q+1)}{(q+1)^2}.\]
We have
\[f'(x)= -\frac{2(k-x)(k+x+1)}{(x+1)^3}-
\frac{2x+1}{(x+1)^2}<0,\qquad \forall x\in [0,k],\]
and hence
\[\sup_{q=0,\ldots, k-2}\abs{a_{q,0}}=\abs{a_{q^*,0}},\qquad q^*:=\max
\setst{q}{f(q)\geq 1}.\] 
The equation $f(x)=1$ is
\[(k-x)(k+x+1)=(x+1)^2,\]
which has a unique positive solution $\eta k$ with  $\eta\in [1/2, 1].$
Using Stirling's approximation, we find there exists $C>0$ so that
\begin{align*}
  \sup_{q=0,\ldots k-2} \abs{a_{q,0}} &=O\lr{
                                    \frac{\lr{k(1+\eta)}!}{\lr{(k\eta)!}^2\lr{k(1-\eta)}!}}=O\lr{\frac{1}{k}C^k}.
\end{align*}
Hence, to prove \eqref{E:small-terms}, it remains to establish the
estimate
\begin{equation}\label{E:small-terms2}
  \abs{ b_q}= O\lr{(1+\ell^2) N^{k-2}2^{k-q}}.
\end{equation}
To do this, write 
\[\abs{b_q}=\lr{\frac{N}{2}}^k\underbrace{\prod_{j=0}^{q-1}\lr{1-\frac{\ell}{N}
  - \frac{2j}{N}}}_{=:g_1(\ell/N)}\underbrace{\prod_{j=q+1}^k\lr{\frac{\ell
    +\frac{d-2}{2} + j}{N}}}_{=:f_1(\ell/N)}.\]
Since $q\leq k-2,$ we have
\[f_1(\ell/N) = \frac{\lr{\ell + \frac{d-2}{2}+k}\lr{\ell
    +\frac{d-2}{2} + k-1}}{N^2} \prod_{j=q+1}^{k-2}\lr{\frac{\ell + \frac{d-2}{2}+j}{N}}.\]
Next, since $\ell \leq N$ and $k\leq N/2,$ we have
\[\prod_{j=q+1}^{k-2}\lr{\frac{\ell + \frac{d-2}{2}+j}{N}} \leq
2^{k-q-2}.\]
Observing that
\[\frac{\lr{\ell + \frac{d-2}{2}+k}\lr{\ell
    +\frac{d-2}{2} + k-1}}{N^2} = O\lr{k^2
  \frac{(1+\ell^2)}{N^2}}\]
confirms \eqref{E:small-terms2} and completes the proof of \eqref{E:small-terms}.
For the remaining two terms, we write
\begin{align*}
a_{k,\beta}b_k + a_{k-1,\beta}b_{k-1}  & = \frac{1}{k!} \cdot \frac{(-1)^k
                                 (2k)!}{(\beta+1)_k} (-N')_k  +
                                 \frac{1}{(k-1)!} \cdot
                                 \frac{(-1)^{k-1}
                                 (2k-1)!}{(\beta+1)_{k-1}} (-N')_{k-1}(\alpha+k)\\
& = \frac{(2k)!}{k! (\beta+1)_k}\lr{\frac{N}{2}}^k
  \tilde{g}\lr{\frac{\ell}{N}}\lr{1 + \frac{\ell}{N}\lr{\beta+k-1}
  +T(\beta,k,N)} .
\end{align*}
Here $T(\beta,k,N)=O\lr{\frac{k^2}{N}}$ and is independent of $\ell$, and
\[\tilde{g}(x):=\prod_{j=0}^{k-2}\lr{1-x - \frac{2j}{N}} .\]
 We have
\[\tilde{g}(\ell/N) = \tilde{g}(0) + \frac{\ell}{N}\tilde{g}'(0) + O\lr{\frac{\ell^2}{N^2}
\sup_{x\in [0,\ell/N]}\abs{\tilde{g}''(x)}},\]
where
\begin{align*}
\tilde{g}(0) = 1 + O(k/N), \qquad \tilde{g}'(0) = (1-k)\lr{1+O(k/N)},
\end{align*} 
and
\[\sup_{x\in [0,\ell/N]}\abs{\tilde{g}''(x)} = \sup_{x\in
  [0,\ell/N]}\abs{\sum_{0\leq j_1< j_2\leq k-2} \prod_{j\neq j_1,
    j_2}\lr{1-\frac{\ell}{N} - \frac{2j}{N}}}=O(k^2).\]
Putting this all together, we obtain,
\begin{align*}
& (\alpha +1)_k\pFq{3}{2}{-k, \,k+1,\,-N'
}{\beta+1, \alpha +1}{1} \\
&=\lr{\frac{N}{2}}^k\cdot \left\{\frac{(2k)!}{k!
                       (\beta+1)_k}\lr{1+ \beta\cdot \frac{\ell}{N}+ T(\beta,k,
                       N)}+ O\lr{\frac{(1+\ell^2)}{N^2}C_3^k}\right\}
\end{align*}
as required. \qed

\section{Proof of Proposition \ref{P:step1}}\label{S:step1-proof}
To prove Proposition \ref{P:step1}, we begin with the following result, which allows us to
replace $E_{\ell,n}^V(\ep)$ by a finite number (depending on
$n$) of its jets at $\ep=0.$
\begin{proposition}\label{P:step1.1}
For any
$V\in L^\infty(\R_+)$ with $\norm{V}_{L^\infty}=1$, $n$ such that $\hbar_{n}<1$, and any $J=J(n)$
satisfying $\liminf_{n\gives \infty} J(n)/\log n=\infty,$
we have
\begin{equation}
  \sup_{\substack{
      \ell\leq n,\, \ell\equiv n \text{ (mod 2)}\\ \ep\in [0,\frac{1}{5}]}}\abs{E_{\ell,n}^V(\ep) - \sum_{j=0}^J \frac{\lr{\hbar_n\ep}^j}{j!}
\frac{d^j}{d\ep^j}\bigg|_{\ep=0} E_{\ell,n}^{V}(\ep)} = O\lr{ \hbar_n^\infty }.
\end{equation}
\end{proposition}

\begin{proof}
  Applying Lemma \ref{L:level-spacing} with $W(r) = \ep\hbar_nV(r^2)$, we find that for
  every $\ep\in [0,1/5]$
    \begin{equation}\label{E:spacing-est}
      \sup_{
      \ell\leq n,\, \ell\equiv n \text{ (mod 2)}}
  d\lr{E_{\ell, n}^V(\ep), \, \Spec(\Op_{\hbar_n,
      \ell}(\ep))\backslash \set{ E_{\ell, n}^V(\ep)}} > \frac{\hbar_n}{2},
    \end{equation}
where $d(x, A)$ denotes the distance from a point $x$ to a set $A.$ As explained in \S \ref{S:perturbation-theory}, we have
\[\frac{d^j}{d\ep^j}E_{\ell, n}^V(\ep)= (-1)^{j-1}\inprod{V
  \psi_{\hbar_n, \ell,n}(\ep)}{[G_\ell(\ep)]^{(j-1)}(\psi_{\hbar_n,\ell, n})},\]
where
\begin{equation}\label{E:Gell-def}
G_\ell(\ep)= \lr{\Op_{\hbar_n, \ell}-E_{\ell, n}^V(\ep)}^{-1}\circ
\Pi_{\psi_{\hbar_n, \ell, n}^{\perp}}\circ V,
\end{equation}
and $V$ denotes multiplication by the function $V(r^2).$ Hence, using
\eqref{E:spacing-est} and that $\norm{V}_{L^\infty}\leq1$,
we find
\begin{equation}\label{E:Gell-norm-est}
\sup_{\substack{
      \ell\leq n,\, \ell\equiv n \text{ (mod 2)}\\ \ep\in [0,1/5]}}
  \norm{G_\ell(\ep)}\leq \frac{2}{\hbar_n}\quad \Rightarrow\quad  \sup_{\substack{
      \ell\leq n,\, \ell\equiv n \text{ (mod 2)}\\ \ep\in
      [0,1/5]}}\abs{\frac{d^j}{d\ep^j}E_{ \ell,
      n}^V(\ep)}\leq \lr{\frac{2}{\hbar_n}}^j.
\end{equation}
Applying Taylor's theorem then gives
\[ \sup_{\substack{
      \ell\leq n,\, \ell\equiv n \text{ (mod 2)}\\ \ep\in [0,1/5]}}\abs{E_{\ell,n}^V(\ep) - \sum_{j=0}^J \frac{\lr{\hbar_n\ep}^j}{j!}
\frac{d^j}{d\ep^j}\bigg|_{\ep=0} E_{ \ell,n}^{V}(\ep)} = O\lr{\lr{2^J\lr{J+1}!}^{-1}}=
O(\hbar_n^\infty)\]
since $(J+1)!\geq e^{\lr{-\log \hbar_n}^2}.$
\end{proof} 
To complete the proof of Proposition \ref{P:step1}, it remains to
check that, provided $\hbar_n<1$ and $J(n),K(n)$ satisfy \eqref{E:KJ-assumptions}, we have
\begin{equation}\label{E:step1-goal1}
  \sup_{\substack{
      \ell\leq n,\, \ell\equiv n \text{ (mod 2)}\\ 0\leq j \leq J}}\hbar_n^j\abs{\frac{d^j}{d\ep^j}\bigg|_{\ep=0} E_{ \ell,n}^{V}(\ep) - \frac{d^j}{d\ep^j}\bigg|_{\ep=0} E_{\ell,n}^{V_K}(\ep)} = O(\hbar_n^\infty).
\end{equation}
To prove this estimate, we again use
\[\frac{d^j}{d\ep^j}\bigg|_{\ep=0} E_{\ell,n}^{V}(\ep)= (-1)^{j-1} \inprod{V\psi_{\hbar_n, \ell, n}}{G_{\ell}^{(j-1)}(\psi_{\hbar_n, \ell,n})},\]
where 
\[G_{\ell}=G_{\ell}(0)=\lr{\Op_{\hbar_n, \ell}(0)-E}^{-1}\circ
\Pi_{\psi_{\hbar_n, \ell, n}}^{\perp}\circ V.\]
Setting for each $K\geq 1$, 
\[G_{\ell,K}:=\lr{\Op_{\hbar_n, \ell}(0)-E}^{-1}\circ
\Pi_{\psi_{\hbar_n, \ell, n}}^{\perp}\circ V_K,\]
we have
\[\frac{d^j}{d\ep^j}\bigg|_{\ep=0} E_{\hbar_n,
  n, \ell}^{V_K}(\ep)= (-1)^{j-1} \inprod{V_K\psi_{\hbar_n, \ell,
    n}}{G_{\ell,K}^{(j-1)}(\psi_{\hbar_n, \ell,n})}\] 
and \eqref{E:step1-goal1} reduces to showing that for each $j\leq J$
and every $\ell\leq n,\, \ell\equiv n \text{ (mod 2)}$
\begin{equation}\label{E:step1-goal2}
 \abs{\hbar_n^j\lr{\inprod{V \psi_{\hbar_n, \ell,
        n}}{G_{\ell}^{(j-1)}(\psi_{\hbar_n, \ell, n})} - \inprod{V_K\psi_{\hbar_n, \ell,
        n}}{G_{\ell,K}^{(j-1)}(\psi_{\hbar_n, \ell, n})}}} = O(\hbar_n^\infty),
\end{equation}
with the implied constant independent of $j,\ell,n.$ We will establish
\eqref{E:step1-goal2} by induction with the
help of the following lemma. 

\begin{Lem}\label{L:truncate-potential}
Suppose $\hbar_n<1$ and $J=J(n),\, K=K(n)$ satisfy \eqref{E:KJ-assumptions}. Then, there exists a constant $C_1>0$ so that if $V$is a $\delta-$slowly varying potential for the energy $E$, with $\delta\in(0,(C_1E)^{-1})$, 
\begin{equation}
  \sup_{\substack{ 
      \abs{m-n}\leq\frac{n}{2}\\ \norm{X}=1}} \abs{\inprod{\lr{V -V_K}
      \psi_{\hbar_n, \ell, m}} {X}} = O\lr{\hbar_n^\infty}.
\end{equation} 
\end{Lem}
\begin{proof}
Let $\chi(r)$ be an auxiliary cut-off function that equals $1$ for
$r\leq \sqrt{4E}$ and $0$ otherwise. Then, since $V$ is bounded, by 
the exponential decay \eqref{E:Agmon} of $\psi_{\hbar_n, \ell, m}$
\begin{equation} \label{E:truncate-potential1}
 \sup_{\substack{ 
      \abs{m-n}\leq\frac{n}{2}\\ \norm{X}=1}} \abs{\inprod{\lr{1-\chi}V
      \psi_{\hbar_n, \ell, m}} {X}} = O\lr{\hbar_n^\infty}.
\end{equation}
Using the definition \eqref{E:slowly-varying} that $V$ is
$\delta-$slowly varying on the support of 
$\chi$ and the assumption $\lim\sup_{n\to\infty}K(n)/\log n = \infty$,
we have, for all $\delta E$  sufficiently small
\begin{equation} \label{E:truncate-potential2}
 \sup_{\substack{ 
      \abs{m-n}\leq\frac{n}{2}\\ \norm{X}=1}} \abs{\inprod{\chi\lr{V-V_K}
      \psi_{\hbar_n, \ell, m}} {X}} = O\lr{\hbar_n^\infty}.
\end{equation}
Finally, again using the exponential decay \eqref{E:Agmon} of $\psi_{\hbar_n, \ell, m}$ and  $\lim\inf_{n\to\infty}K(n)/ n = 0,$ we obtain
\begin{equation} \label{E:truncate-potential3}
 \sup_{\substack{
      \abs{m-n}\leq\frac{n}{2}\\ \norm{X}=1}} \abs{\inprod{\lr{1-\chi}V_K
      \psi_{\hbar_n, \ell, m}} {X}} = O\lr{\hbar_n^\infty},
\end{equation}
which completes the proof.
\end{proof}

To prove \eqref{E:step1-goal2} by induction, note that Lemma
\ref{L:truncate-potential} is precisely the base case $j=1.$ Next,
suppose we have already shown \eqref{E:step1-goal2} for some $j\geq
1.$ Then, using Lemma \ref{L:truncate-potential} and the norm estimate
from \eqref{E:Gell-norm-est}, we have
\begin{equation}\label{E:V-to-VK}
\hbar_n^{j}\inprod{V\psi_{\hbar_n, \ell, n}}{G_{\ell}^{(j)}(\psi_{\hbar_n,
    \ell,n})}=\hbar_n^{j}\inprod{V_K\psi_{\hbar_n, \ell,
    n}}{G_{\ell}^{(j)}(\psi_{\hbar_n, \ell,n})}+ O(\hbar_n^\infty).
\end{equation}
The adjoint of $G_{\ell}$ is 
\[G_{\ell}^*=V\circ \lr{\Op_{\hbar_n, \ell}- E}^{-1}\circ
\Pi_{\psi_{\hbar_n, \ell, n}}^{\perp}\]
and hence
\[G_{\ell}^*\lr{V_K \psi_{\hbar_n,\ell, n}}=
\hbar_n^{-1}\sum_{\substack{m\text{ s.t. } \abs{m-n}\leq 2K\\ m\neq n}}\frac{\inprod{V_K
  \psi_{\hbar_n,\ell, m}}{\psi_{\hbar_n,\ell, n}}}{m-n}
\psi_{\hbar_n, \ell, m}.\]  
The sum in the previous line is truncated to $\abs{m-n}\leq 2K$ since
by Proposition \ref{P:step2}, the numerator vanishes
unless $\abs{m-n}\leq 2K.$ To complete the proof, we write
\begin{align*}
 & \hbar_n^j \inprod{V\psi_{\hbar_n, \ell, n}}{G_{\ell}^{(j)}(\psi_{\hbar_n,
    \ell,n})} = \hbar_n^j \inprod{V_K \psi_{\hbar_n, \ell, n}}{G_\ell^{(j)}(\psi_{\hbar_n, \ell, n})}+O(\hbar^\infty)\\
&= \sum_{\substack{m\text{ s.t. } \abs{m-n}\leq 2K\\ m\neq n}}\frac{\inprod{V_K
  \psi_{\hbar_n, \ell,m}}{\psi_{\hbar_n,\ell,n}}}{m-n}
\hbar^{j-1} \inprod{V\psi_{\hbar_n, \ell, m}}{G_{\ell}^{(j-1)}\psi_{\hbar_n,
  \ell, n}} + O\lr{\hbar_n^\infty}\\
&= \sum_{\substack{m\text{ s.t. } \abs{m-n}\leq 2K\\ m\neq n}}\frac{\inprod{V_K
  \psi_{\hbar_n,\ell, m}}{\psi_{\hbar_n, \ell,n}}}{m-n}
\hbar^{j-1} \inprod{V_K\psi_{\hbar_n, \ell, m}}{G_{\ell,K}^{(j-1)}\psi_{\hbar_n,
  \ell, n}} + O\lr{\hbar_n^\infty}\\
&=\hbar_n^j \inprod{V_K\psi_{\hbar_n, \ell, n}}{G_{\ell,K}^{(j)}(\psi_{\hbar_n,
    \ell,n})}+O\lr{\hbar_n^\infty},
\end{align*}
where in the second-to-last line we used \eqref{E:V-to-VK} the
inductive hypothesis and the fact that $\limsup_{n\gives \infty}\frac{K(n)}{n}=0.$

\bibliography{mybib}{}

\begin{thebibliography}{10}

\bibitem{beck2016nodal}
Thomas Beck, Spencer~T Becker-Kahn, and Boris Hanin.
\newblock Nodal sets of smooth functions with finite vanishing order and
  p-sweepouts.
\newblock {\em arXiv preprint arXiv:1604.04307}, 2016.

\bibitem{berard2014number}
Pierre B{\'e}rard and Bernard Helffer.
\newblock On the number of nodal domains of the 2d isotropic quantum harmonic
  oscillator--an extension of results of {A.} {Stern}.
\newblock {\em arXiv preprint arXiv:1409.2333}, 2014.

\bibitem{berard2015nodal}
Pierre B{\'e}rard and Bernard Helffer.
\newblock On the nodal patterns of the 2d isotropic quantum harmonic
  oscillator.
\newblock {\em arXiv preprint arXiv:1506.02374}, 2015.

\bibitem{berard2017some}
Pierre B{\'e}rard and Bernard Helffer.
\newblock Some nodal properties of the quantum harmonic oscillator and other
  {Schr{\"o}dinger} operators in {$R^2$}.
\newblock {\em arXiv preprint arXiv:1506.02374}, 2017.

\bibitem{bies2002nodal}
WE~Bies and EJ~Heller.
\newblock Nodal structure of chaotic eigenfunctions.
\newblock {\em Journal of Physics A: Mathematical and General}, 35(27):5673,
  2002.

\bibitem{canzani2016nodal}
Yaiza Canzani and John~A Toth.
\newblock Nodal sets of {Schr{\"o}dinger} eigenfunctions in forbidden regions.
\newblock In {\em Annales Henri Poincar{\'e}}, volume~17, pages 3063--3087.
  Springer, 2016.

\bibitem{donnelly1988nodal}
Harold Donnelly and Charles Fefferman.
\newblock Nodal sets of eigenfunctions on {Riemannian} manifolds.
\newblock {\em Inventiones mathematicae}, 93(1):161--183, 1988.

\bibitem{erdelyi2010asymptotic}
Arthur Erdelyi.
\newblock Asymptotic expansions.
\newblock 2010.

\bibitem{gichev2009some}
Victor Gichev.
\newblock Some remarks on spherical harmonics.
\newblock {\em St. Petersburg Mathematical Journal}, 20(4):553--567, 2009.

\bibitem{hanin2014nodal}
Boris Hanin, Steve Zelditch, and Peng Zhou.
\newblock Nodal sets of random eigenfunctions for the isotropic harmonic
  oscillator.
\newblock {\em International Mathematics Research Notices},
  2015(13):4813--4839, 2014.

\bibitem{hanin2017scaling}
Boris Hanin, Steve Zelditch, and Peng Zhou.
\newblock Scaling of harmonic oscillator eigenfunctions and their nodal sets
  around the caustic.
\newblock {\em Communications in Mathematical Physics}, 350(3):1147--1183,
  2017.

\bibitem{hardt1999critical}
R~Hardt, M~Hoffmann-Ostenhof, T~Hoffmann-Ostenhof, N~Nadirashvili, et~al.
\newblock Critical sets of solutions to elliptic equations.
\newblock {\em J. Differential Geom}, 51(2):359--373, 1999.

\bibitem{hardt1987nodal}
Robert Hardt, Leon Simon, et~al.
\newblock {\em Nodal sets for solutions of elliptic equations}.
\newblock Centre for Mathematical Analysis, ANU, 1987.

\bibitem{Heller-gallery}
Eric Heller.
\newblock Eric {J.} {Heller} gallery.
\newblock {\em
  \url{http://ejheller.jalbum.net/Eric\%20J\%20Heller\%20Gallery/slides/Nodal8.html}}.
\newblock Accessed: 2017-07-16.

\bibitem{hoffmann1988asymptotics2}
M~Hoffmann-Ostenhof.
\newblock Asymptotics of the nodal lines of solutions of 2-dimensional
  schr{\"o}dinger equations.
\newblock {\em Mathematische Zeitschrift}, 198(2):161--179, 1988.

\bibitem{hoffmann1988asymptotics}
Maria Hoffmann-Ostenhof and Thomas Hoffmann-Ostenhof.
\newblock On the asymptotics of nodes of l2-solutions of schr{\"o}dinger
  equations in dimensions {$\geq 3$}.
\newblock {\em Communications in mathematical physics}, 117(1):49--77, 1988.

\bibitem{hoffmann1986continuity}
Maria Hoffmann-Ostenhof, Thomas Hoffmann-Ostenhof, and J{\"o}rg Swetina.
\newblock Continuity and nodal properties near infinity for solutions of
  2-dimensional schr{\"o}dinger equations.
\newblock {\em Duke mathematical journal}, 53(1):271--306, 1986.

\bibitem{hoffmann1987asymptotics}
Maria Hoffmann-Ostenhof, Thomas Hoffmann-Ostenhof, and J{\"o}rg Swetina.
\newblock Asymptotics and continuity properties near infinity of solutions of
  schr{\"o}dinger equations in exterior domains.
\newblock {\em Ann. Inst. H. Poincar{\'e}46}, pages 247--280, 1987.

\bibitem{jin2017semiclassical}
Long Jin.
\newblock Semiclassical {Cauchy} estimates and applications.
\newblock {\em Transactions of the American Mathematical Society},
  369(2):975--995, 2017.

\bibitem{suslov2008hahn}
Sergei~K Suslov and Benjamin Trey.
\newblock The {Hahn} polynomials in the nonrelativistic and relativistic
  {Coulomb} problems.
\newblock {\em Journal of Mathematical Physics}, 49(1):012104, 2008.

\bibitem{Taylor-dispert}
Michael Taylor.
\newblock Self-adjoint perturbations in the discrete case.
\newblock {\em \url{http://www.unc.edu/math/Faculty/met/dispert.pdf}}.
\newblock Accessed: 2017-07-16.

\end{thebibliography}
\bibliographystyle{plain}

\end{document}